\documentclass[secthm,seceqn,amsthm,ussrhead,11pt]{amsart}
\usepackage{amsmath,latexsym}
\usepackage[psamsfonts]{amssymb}
\usepackage{times}
\usepackage[mathcal]{euscript}
\usepackage[cp1251]{inputenc}

\numberwithin{equation}{section}  \tolerance=6000

\renewcommand{\epsilon}{\varepsilon}

\newcommand{\be}{\begin{equation}}
\newcommand{\ee}{\end{equation}}

\newcommand{\C}{\mathbb{C}}

\newcommand{\R}{\mathbb{R}}

\newcommand{\T}{\mathbb{T}}

\newcommand{\Z}{\mathbb{Z}}

\newcommand{\cE}{{\mathcal E}}

{\bf}{\it}
\newtheorem{theorem}{Theorem}[section]
\newtheorem{lemma}[theorem]{Lemma}
\newtheorem{corollary}[theorem]{Corollary}

\newtheorem{remark}[theorem]{Remark}

\date{\today}

\begin{document}
\title[Asymptotics of Eigenvalues of the Two-particle Schr\"{o}dinger
operators on lattices] {Asymptotics of Eigenvalues of the
Two-particle Schr\"{o}dinger operators on lattices. \\}

\author[Saidakhmat  N. Lakaev,\,
  Shohruh Yu. Holmatov]{Saidakhmat  N. Lakaev,\,
  Shohruh Yu. Holmatov}

\address{
{$^2$ Samarkand State University, Samarkand (Uzbekistan)} \ {E-mail:
slakaev@mail.ru }}

\address{
$^3$ {Samarkand State University, Samarkand (Uzbekistan)}{
E-mail:Shohruhhon1@mail.ru}}

\begin{abstract}
The Hamiltonian of a system of two quantum mechanical particles
moving on the $d$-dimensional lattice $\Z^d$ and interacting via
zero-range attractive pair potentials is considered. For the
two-particle energy operator $H_{\mu}(K),$ $K\in \T^d=(-\pi,\pi]^d$
-- the two-particle quasi-momentum, the existence of a unique
positive eigenvalue $z(\mu, K)$  above the upper edge of the
essential spectrum of $H_{\mu}(K)$ is proven and asymptotics for
$z(\mu, K)$ are found when $\mu$ approaches to some $\mu_0(K)$ and
$K\to 0.$
\end{abstract}

\maketitle



\section{ Introduction}

In this paper we will consider the family of the two-particle
Schr\"{o}dinger operators associated to a system of two identical
particles moving on $d$ - dimensional cubic lattice $\Z^d$ and
interacting via zero-range potentials. In the momentum
representation the cor\-res\-ponding operator is of the form
$$H_{\mu}(K) = H_0(K)+\mu V,\quad K\in\T^d,$$ where $\T^d$ -- $d$
dimensional torus. The non perturbed operator $H_0(K)$ is the
multiplication ope\-ra\-tor by the function $$
\cE_K(q)=\epsilon\left(\frac K2+q\right)+\epsilon\left(\frac
K2-q\right),$$ where
\begin{equation}\label{rrrr}
\epsilon(q)=\sum\limits_{i=1}^d(1-\cos q_i),
\end{equation}
$V$ is integral operator of rank one and $\mu>0$ is repulsive
interaction. $H_{\mu}(K)$ has continuous spectrum $[\cE_{\min
}(K),\cE_{\max }(K)]$ and at most one eigenvalue $z(\mu,K)$ on the
right from $\cE_{\max}(K).$

In celebrated work  \cite{MK.BS} of B.Simon and M.Klaus it is
considered a family of Schr\"{o}dinger ope\-ra\-tors
$H=-\Delta+\lambda V$ and, a situation where as $\lambda \downarrow
\lambda_0$ some eigenvalue $e_i(\lambda)\uparrow 0,$ i.e., as
$\lambda \downarrow \lambda_0$ an eigenvalue is absorbed into
continuous spectrum, and conversely, as $\lambda \uparrow
\lambda_0+\epsilon$ continuous spectrum "gives birth" to a new
eigenvalue. This phenomenon in \cite{MK.BS} is called "coupling
constant threshold".

The phenomenon coupling constant threshold is a significant tool not
only for the two-particle continuous Hamiltonians (see
\cite{MK.BS}), but also in the existence of the three-particle bound
states of the Hamiltonians of a system of three particles, in
particular, for the Efimov effect (see \cite{{Sob},{Yaf2},{Tam1}}).

In the case of the three-particle lattice Schr\"{o}dinger
operators $\textit{\textbf{H}}_{\mu_0} (K),$ $K\in \T^3$ --
three-particle quasi-momentum, associated to the Hamiltonian of a
system of three particles on $\Z^3$ interacting via zero-range
pair potentials $\mu<0$ the following phenomenon is also deeply
related to the coupling constant threshold $\mu_0<0$ of the
two-particle Schr\"{o}dinger operators: for $\mu=\mu_0$ the
corresponding three-particle lattice Schr\"{o}dinger operator
$\textit{\textbf{H}}_{\mu_0}(0)$ has infinitely many eigenvalues,
whereas $\textit{\textbf{H}}_{\mu_0} (K),$ $K\ne 0$ has only
finitely many (see \cite{{SL},{ALZM}}).

Throughout physics, stable composite objects are usually formed by
way of attractive forces, which allow the constituents to lower
their energy by binding together. Repulsive forces separate
particles in free space. However, in structured environment such as
a periodic potential and in the absence of dissipation, stable
composite objects can exist even for repulsive interactions (see
\cite{RBAP}).

In the present paper, for the two-particle ope\-ra\-tor
$H_{\mu}(K),$ $K\in\T^d,$ it is established: at first, if $d\ge 3$
then there exists such $\mu_0=\mu_0(K)>0$ (coupling constant
threshold) that the operator has non eigenvalue for any
$0<\mu<\mu_0,$ but for any $\mu>\mu_0$ there is a unique
eigenvalue $z(\mu,K)$ of $H_{\mu}(K)$ lying above the upper edge
of $\sigma_{\rm ess}(H_{\mu}(K)).$ In \cite{MK.BS}, it is only
assumed the existence of the coupling constant threshold
$\lambda_0>0,$ but in this work, the coupling constant threshold
is definitely found by the given data.

Secondly, as in \cite{MK.BS}, two questions will concern us: (i)
For fixed $K\in\T^d,$ is $z(\mu,K)$ analytic at $\mu=\mu_0,$ if
singular, does it have an expansion in some singular quality like
$(\mu-\mu_0)^{\alpha},$ $\alpha\in\R?$ (ii) For fixed $K\in\T^d,$
what is the rate at which $z(\mu,K)$ approaches to the upper edge
of the essential spectrum $\cE_{\max }(K),$ as $\mu$ approaches to
$\mu_0(K)?$

Thirdly, lattice Schr\"{o}dinger operators are strictly dependant
on the quasi-momentum $K$ of the system, and it is one of the
differences between continuous and lattice operators. Thus, we are
interested in one more question: (iii) What is the rate at which
$z(\mu_0(0),K)$ approaches to $\cE_{\max }(0)$ as $K\to 0?$

The paper is organized as follows.

In Section \ref{2} we introduce the two-particle operator $H_\mu(K)$
and give a location of its essential spectrum.

In Section \ref{3} we define coupling constant threshold $\mu_0(K)$
in some concrete domain $\Pi_0\subset\T^d$ and give main results of
the paper.

In Section \ref{4} we prove main results.

In Appendix for reader's convenience,  we give full proves of some
consequences of the implicit function theorem, used in the proof
of Theorem \ref{assymptotica} and proves of some lemmas, used in
the proof of Lemma \ref{asimp}.

\section{ The two-particle operator
$H_\mu(K)$ and its Essential spectrum} \label{2}

Let $\Z^d$ be $d$ dimensional hybercubic lattice and
$\T^d=(\R/2\pi\Z)^d=(-\pi,\pi]^d$ be $d$ dimensional torus
(Brillouin zone), the dual group of $\Z^d.$ Denote by
$L_2({\T}^d,d\eta)$ the Hilbert space of square-integrable functions
on $\T^d,$ where $d\eta$ -- Haar measure in $\T^d.$

Let $L^{e}_2 ({\T}^d)$ the subspace of even functions in
$L_2({\T}^d,d\eta).$ Consider the analytic function on $\T^d$
$$\epsilon(q)=\sum\limits_{i=1}^d (1-\cos q_i).$$

In the momentum representation the two-particle Hamiltonians are
given by the bounded self-adjoint operators on the Hilbert space
$L^e_2({\T}^d)$ as following (see \cite{ALMZM}):
\begin{equation}
 H_\mu(K)=H_0(K)+\mu V.
\end{equation}

The non-perturbed operator $H_0(K)$ on $L^{e}_2({\T}^d)$ is
multiplication operator by the function $\cE_K(\cdot):$
\begin{equation}
(H_{0}(K)f)(q)=\cE_K( q)f(q),\quad f \in L^{e}_2({\T}^d),
\end{equation} where
 \begin{equation} \label{kinetik}
\cE_{K}(q)=\frac{1}{m}\left[ \varepsilon \left(\frac K2-q\right)
+\varepsilon \left(\frac K2+q\right)\right] = \frac{1}{m}
\sum\limits_{j=1}^{d}\left[2 - 2\cos{\left(\frac{K_j}{2}\right)}\cos
q_j\right],
\end{equation}

The perturbation $V$ is an integral operator of rank one
$$(Vf)(p)=\int\limits_{{\T}^d} f(q)d\eta,\quad f\in L_2 ({\T}^d).$$

Further without loss of generality we assume that $m=1.$

The perturbation $V$ of the multiplication operator $H_0(K)$ is a
self-adjoint operator of rank one. Therefore in accordance to Weil's
theorem (see \cite{RSIV}), the essential spectrum of $H_{\mu}(K), \,
K\in {\T}^{d}$ fills the following interval on the real axis:
$$
  \sigma_{ess}(H_{\mu}(K))=\sigma(H_{0}(K))=
  [\cE _{\min }(K), \cE_{\max}(K)],
$$ where
\begin{gather*}
\cE _{\min }(K)\equiv\min_{p}\cE_K(p)=\sum\limits_{j=1}^d
\left[2-2\cos{\left(\frac{K_j}{2}\right)}\right],\\
\cE_{\max}(K)\equiv\max_{p}\cE_K(p)
=\sum\limits_{j=1}^d \left[2+2\cos{\left(\frac{K_j}{2}\right)}\right].\\
\end{gather*}

\begin{remark}
We remark that the essential spectrum
$[\cE_{\min}(K),\cE_{\max}(K)]$ strongly depends on the
quasi-momentum $K\in\T^d;$ when $K=\vec{\pi}=(\pi,\pi,...,\pi)\in
\T^d$ the essential spectrum of $H_{\mu}(K)$ degenerated to the set
consisting of a unique point $\{\cE_{\min}(\vec{\pi})=
\cE_{\max}(\vec{\pi}) =2d\}$ and hence  the essential spectrum of
$H_{\mu}(K)$ is not absolutely continuous for all $K\in \T^d.$
\end{remark}

\section{Main Results}\label{3}

Set
$$\Pi_n=\left\{k\in\T^d:\hspace{1mm} \text{at most $n$ coordinates
of $k$ is equal to $\pi$}\right\},\quad 0\le n\le d.$$ It is clear
that $\Pi_d=\T^d,$ $\Pi_m\subset\Pi_n$ if $m<n$ and
$$\Pi_0=\left\{k=(k_1,\ldots,k_d\in\T^d:\, k_i\ne \pi,\,
i=1,\ldots,d)\right\}.$$

Let $\C$ be the complex plane. For any $K \in \T^d$ we define a
regular function $\nu(K,\cdot)$ in ${\C} {\setminus}
[\cE_{\min}(K),\cE_{\max}(K)]$ by
\begin{equation}\label{eeee}
  \nu(K,z)=\int\limits_{{\T}^d}
  \frac{d\eta} {z-\cE_K(q)}.
\end{equation}

For $d\geq 3$ and $K\in\Pi_{d-3},$ the function $\cE_K(q)$ has a
unique non degenerated maximum and, consequently, the following
integral exists and defines analytic function on $\Pi_0:$
\begin{equation}\label{qqqq}
\nu(K) =\nu(K,\cE_{\max}(K))=
 \int\limits_{{\T}^d}
  \frac{d\eta} {\cE_{\max}(K)-\cE_K(q)}
\end{equation}
(see Lemma \ref{asimp}).

\begin{remark}\label{deguniq}
If $n$ coordinates of $K$ equals to $\pi$ then
$\cE_{\max}(K)-\cE_K(q)$ in \eqref{qqqq} can be considered as the
function defined on $\T^{d-n}$ having non-degenerated maximum at
$\vec{\pi}=(\pi,\ldots,\pi)\in\T^{d-n}.$ Therefore without loss of
generality we can always assume that any coordinates of $K\in\T^d$
is not equal to $\pi,$ that is, $K\in\Pi_0\subset\T^d.$
\end{remark}

Let $d\geq3.$ Define the positive number $\mu_0(K),$ $K\in\Pi_{0}$
as
\begin{equation}
\mu_0(K)=\frac{1}{\nu(K)}=\left(\int\limits_{\T^d}
\frac{d\eta}{\cE_{\max}(K)-\cE_K(q)}\right)^{-1}.
\end{equation}

\vspace{3mm}

The following theorem is on the existence and dependance of
eigenvalues of the two-particle operator $H_{\mu}(K)$ on
interaction energy $\mu>0.$ In fact, we prove that there exists a
unique eigenvalue $z(\mu,K)$ of $H_{\mu}(K),$ $K\in\Pi_0$ above
its essential spectrum depending on the coupling constant
$\mu_0=\mu_0(K)>0.$ Moreover we find an expansion for the
difference
$$z(\mu,K)-\cE_{\max}(K)$$ at $\mu=\mu_0.$ This expansion is highly
dependant on dimension $d\ge 3$ of the quasi-momentum $K\in\T^d$:
$1)$ if $d=3,$ then $z(\mu,K)-\cE_{\max}(K)$ is an analytic
function of $\mu-\mu_0$ with the leading term of $(\mu-\mu_0)^2;$
$2)$ if $d=4,$ then the difference does not expand to Puizo
series, but it has an expansion with the first term of
$\sigma=(\mu-\mu_0)(-\ln (\mu-\mu_0))^{-1};$ $3)$ if $d\ge 5$ and
odd, then $z(\mu,K)-\cE_{\max}(K)$ is and analytic function of
$\alpha=(\mu-\mu_0)^{1/2}$ and the leading term of the expansion
is $\mu-\mu_0;$ $4)$ if $d\ge 6$ and even, then the expansion of
$z(\mu,K)-\cE_{\max}(K)$ is an analytic  function of
$$\sigma =(\mu-\mu_0)^{1/2},\quad \tau=(\mu-\mu_0)^{1/2}\ln
(\mu-\mu_0)^{1/2}$$ with the first term of $\mu-\mu_0.$

\begin{theorem}\label{assymptotica}
Let $d\ge 3.$  Then for any $K\in\Pi_{d-3}$ and $\mu>\mu_0(K)$ the
operator $H_{\mu}(K)$ has a unique eigenvalue $z(\mu, K)$ lying
above the upper edge $\cE_{\max}(K)$ of the essential spectrum of
$H_{\mu}(K).$  Moreover, for any $K\in \Pi_0$ the relation $\mu\to
\mu_0(K)$ holds if and only if $z(\mu,K)\to
z(\mu_0(K),K)=\cE_{\max}(K)$ and for $\mu-\mu_0(K)$ sufficiently
small and positive, the difference $z(\mu,K)-z(\mu_0(K),K)$ has
following absolutely convergent expansions:
\begin{itemize}
\item[(\rm{i})]if $d=3,$ then
$$z(\mu,K)-z(\mu_0(K),K)=\left(\sum\limits_{n=1}^{\infty} c_n(K)
[\mu -\mu_0(K)]^n\right)^2,$$ where $c_n(K),$ $n=1,2,\ldots$ is a
real number with $$c_1(K)=\left(\dfrac{\pi(\mu_0(K))^2
\Phi_0(K)}{2}\right)^{-1}$$ and
$$\Phi_0(K)=\dfrac{c}{\sqrt{\cos \dfrac{K_1}{2}\ldots\cos
\dfrac{K_d}{2}}},\quad c=const;$$
\item[(\rm{ii})]if $d=4$ then  \be\label{dd=4}
z(\mu,K)-z(\mu_0(K),K)=\sum\limits_{n\ge 1,\,m,k,l\ge 0} c(n,m,k,l;
K)\, \sigma^n\,\tau^m\,\omega^k\, \lambda^l \ee with
$$\sigma=\frac{\mu-\mu_0(K)}{-\ln
(\mu-\mu_0(K))},\quad \tau=\frac{1}{-\ln (\mu-\mu_0(K))}, \quad
\omega=\frac{\ln\ln (\mu-\mu_0(K)^{-1}}{-\ln (\mu-\mu_0(K))},\quad
\lambda =\mu-\mu_0(K)$$ and $c(n,m,k,l; K),$ $n=1,2,\ldots,$
$m,k,l=0,1,2\ldots$ is a real number with
$$c(1,0,0,0; K)=\frac{2}{(\mu_0(K))^2\Phi_0(K)};$$
\item[(\rm{iii})] if $d\ge 5$ and odd, then $$z(\mu,K)-
z(\mu_0(K),K)=\left(\sum\limits_{n\ge 1} c_n(K)
(\mu-\mu_0(K))^{n/2}\right)^{2},$$ where $c_n(K),$ $n=1,2,\ldots$ is
a real numbers with $$c_1(K)=\left(-(\mu_0(K))^{2}
\,\,\frac{\partial \nu}{\partial z}(K,\cE_{\max}(K))\right)^{-1/2}$$
and $\nu(\cdot,\cdot)$ is defined by \eqref{eeee};
\item[(\rm{iv})] if $d\ge 6$ and even, then \be\label{dd>6}
z(\mu,K)-z(\mu_0(K),K)=\left(\sum\limits_{n\ge 1,\,m\ge 0} c(n,m;
K)\, \sigma^n\,\tau^m\right)^2\ee with $$\sigma
=(\mu-\mu_0(K))^{1/2},\quad \tau=(\mu-\mu_0(K))^{1/2}\ln
(\mu-\mu_0(K))^{1/2}$$ where $c(n,m; K),$ $n=1,2,\ldots,$
$m=0,1,2,\ldots$ is a real number with
$$c(1,0; K)=\left(-(\mu_0(K))^{2} \,\,\frac{\partial \nu}{\partial
z}(K,\cE_{\max}(K))\right)^{-1/2}.$$
\end{itemize}
\end{theorem}

In the following theorem, we show existence and describe dependance
of eigenvalues of $H_{\mu^0}(K),$ $\mu^0=\mu_0(0)$ on the
quasi-momentum $K:$ for any $K\in\Pi_0\setminus\{0\}$ there exists a
unique eigenvalue $z(\mu^0,K)$ of the operator and we find an
asymptotics of the difference $z(\mu^0,K)-\cE_{\max}(0)$ as $K\to
0.$

\begin{theorem}\label{quasi}
For any $K\in \Pi_0\setminus\{0\},$ the operator $H_{\mu_0(0)}(K)$
has a unique eigenvalue $z(\mu_0(0),K)$ lying above the essential
spectrum $\sigma_{ess}(H_{\mu_0(0)}(K)).$ Moreover for
$z(\mu_0(0),K)$ the following asymptotics hold:
\begin{itemize}
\item[(\rm{i})] if $d=3,$ then, $$z({\mu_0(0)},K)-\cE_{\max}(0)=
-\frac14|K|^2+ O(|K|^4),\; K\to 0;$$

\item[(\rm{ii})] if $d=4,$ then
$$z(\mu_0(0),K)-\cE_{\max}(0)=-\frac14|K|^2+o(|K|^{2}),\; K\to 0;$$

\item[(\rm{iii})] if $d\ge5,$ then
$$z(\mu_0(0),K)-\cE_{\max}(0)=\alpha |K|^2+o(|K|^{2}),\; K\to
0,$$ where $$\alpha=-\left(\cfrac{\partial^2\nu(0)}{\partial
K_1^2}\right)\left(\dfrac{\partial\nu}{\partial
z}(0,\cE_{\max}(0))\right)^{-1} - \frac14,$$ and $\nu(\cdot)$ is
defined by \eqref{qqqq}.
\end{itemize}
\end{theorem}

\section{The proof of the results}\label{4}

For any $K \in \T^d,$ we define the Fredholm determinant
associated to the operator $H_\mu(K)$ as analytic function in
$z{\in } {\C} { \setminus } [\cE_{\min}(K),\cE_{\max}(K)]$ by
\begin{equation*}\label{detir..}
{\Delta}_\mu(K,z) =1-\mu\nu(K,z).
\end{equation*}

Observe that the function ${\Delta}_\mu(K,z)$ is real-analytic in
$\T^d\times ({ \mathbb{R}} {\setminus}
[\cE_{\min}(K),\cE_{\max}(K)]).$

\vspace{5mm}

\begin{lemma}\label{eigen}
For any $K\in\T^d$ the eigenvalues of $H_{\mu}(K)$ outside the
essential spectrum coincides with the zeros of
$\Delta_{\mu}(K,\cdot).$
\end{lemma}

\begin{proof}[\bf Proof]
Let $z\in \C\setminus \sigma_{ess}(H_{\mu}(K))$ be an eigenvalue
of $H_{\mu}(K)$ and $f\in L_2^e(\T^d)$ be one of the eigenvectors
corresponding to $z.$ Then by the definition of eigenvalue one can
get that $$(z-\cE_K(p))f(p)= \mu\int\limits_{\T^d} f(q)d\eta.$$
Therefore, the equation $$\phi(p)=\mu \int\limits_{\T^d}
\frac{\phi(q)d\eta}{z-\cE_K(q)}$$ has a simple solution
$\phi(q)\equiv 1$ (up to constant factor) in $L_2^e(\T^d).$ So we
obtain that $$1=\mu \int\limits_{\T^d} \frac{d\eta}{z-\cE_K(q)}$$
and hence $\Delta_{\mu}(K,z)=0.$

Conversely, let $z\in \C\setminus \sigma_{ess}(H_{\mu}(K))$ be a
solution of the equation $\Delta_{\mu}(K,z)=0.$ Then it is easy to
see that the function $f(\cdot)=(z-\cE_K(\cdot))^{-1}\in
L_2^{e}(\T^d)$ satisfies the equality $H_{\mu}(k)f=z\,f,$ which
means $z$ is eigenvalue (see also \cite{SL}).
\end{proof}

The main results of the paper and their proves are based on the
following Lemma (see \cite{L92}), which is important tool not only
in studying spectral properties of the two-particle
Schr\"{o}dinger operators, but also in the spectral analysis of
the three-particle Schr\"{o}dinger operators (see \cite{SL},
\cite{ALZM}).

\begin{lemma}\label{asimp}
{\rm (i)} If $d\geq 3$ and $K\in\Pi_{0},$ the integral
\begin{equation}\label{ee}
\nu(K) =\nu(K,\cE_{\max}(K))=
 \int\limits_{{\T}^d}
  \frac{d\eta} {\cE_{\max}(K)-\cE_K(q)}
\end{equation}
exists and defines analytic function on $\Pi_{0}\subset\T^d.$
\vspace{3mm}

{\rm (ii)} Let $K\in\Pi_0$ and $z>\cE_{\max}(K).$ The function
${\nu}(K, z)$ can be represented as
\begin{eqnarray*}
\nu(K, z)&= &-\dfrac{\Phi_0(K)}{2} \left(\cE_{\max}(K)-z\right)^m
\ln(z-\cE_{\max}(K))+\\&+ &\left(\cE_{\max}(K)-z\right)^{m+1}
\ln(z-\cE_{\max}(K))\Phi_{11}(K,z)+ \nu(K) +
\Phi_2(K,z),\end{eqnarray*} if $d=2m+2$ and
\begin{eqnarray*}
\nu(K, z)&=& \dfrac{\pi \Phi_0(K)}{2}\hspace{1mm}
\dfrac{\left(\cE_{\max}(K)-z\right)^{m}}{\sqrt{z-\cE_{\max}(K)}}
+\\
&+&\left(\cE_{\max}(K)-z\right)^{m+1/2} \Phi_{12}(K,z)+ \nu(K) +
\Phi_2(K,z),
\end{eqnarray*} if $d=2m+1,$ where $$\Phi_0(K)=\dfrac{c}{\sqrt{\cos
\dfrac{K_1}{2}\ldots\cos \dfrac{K_d}{2}}},\quad c=const,$$ and
$$\Phi_{12}(K,z)=\sum\limits_{l=0}^{\infty}b_l(K)\left(z-\cE_{\max}(K)
\right)^{l/2}$$ and $\Phi_{11}(K,\cdot)$ and $\Phi_2(K,\cdot),$
$K\in \Pi_0$ are analytic functions in some neighborhood
${V}_{\epsilon}(\cE_{\max}(K))$ of $z=\cE_{\max}(K),$
$\Phi_2(K,\cE_{\max}(K))=0$ and $b_l(K),$ $l=0,1,2,\ldots$ are some
real coefficients.
\end{lemma}

\begin{proof}[\bf Proof of Lemma \ref{asimp}] (i) It is easy to see that
$q_0=\vec{\pi}$ is a unique non degenerate maximum of $\cE_K(q).$ We
rewrite function $\nu(K, z)$ as
\begin{gather}\label{kkk}
\nu(K, z)= \int\limits_{U_{\delta}
(\vec{\pi})}\frac{d\eta}{z-\cE_K(q)}+
\int\limits_{\mathbb{T}^d\setminus
U_{\delta}(\vec{\pi})}\frac{d\eta}{z-\cE_K(q)}=\\
=G_1(K, z)+G_2(K, z).\nonumber
\end{gather}

Observe that $G_2(\cdot; z),$ $z>\cE_{\max}(K)$ and
$G_2(K,\cdot),$ $K\in \Pi_0$ are analytic functions on $\Pi_0$ and
$(\cE_{\max}(K), \infty)$ respectively. Moreover observe that
$G_2(K, \cE_{\max}(K))$ is regular on $\Pi_0.$

In the first integral making a change of variables $q=\phi(x),$
where
$$\phi:U_{\delta}(\vec{\pi})\rightarrow W_{\gamma}(0),\quad q_j:\to
\arccos\left(\frac{x_{j}^2}{2\cos\frac K2}-1\right),$$ we get:
$$G_1(K, z)= \int\limits_{W_{\gamma}
(0)}\frac{1}{z-\cE_{\max}(K)+\sum\limits_{j=1}^d x_j^2}\times
$$$$\times\dfrac{2dx_1}{\sqrt{2\cos\cfrac
{K_1}{2}}\sqrt{2-\cfrac{x_{1}^2}{2\cos\cfrac {K_1}{2}}}}\ldots
\dfrac{2dx_d}{\sqrt{2\cos\cfrac
{K_d}{2}}\sqrt{2-\cfrac{x_d^2}{{2\cos\cfrac {K_d}{2}}}}}.$$
\vspace{3mm}

There is no loss of generality in assuming that $W_{\gamma}(0)$ is
sphere in $\R^d$ with center at $x=0$ and with radius $\gamma>0.$

Passing spherical coordinates as
\begin{eqnarray*}
x_1& = & r\cos \psi_1 \cos \psi_2 \ldots \cos \psi_{d-2} \cos \psi_{d-1} \\
x_2& = & r\cos \psi_1 \cos \psi_2 \ldots \cos \psi_{d-2} \sin \psi_{d-1} \\
x_3& = & r\cos \psi_1 \cos \psi_2 \ldots \sin \psi_{d-2}  \\
\vdots \\
x_{d}& = & r\sin \psi_1 \\
\end{eqnarray*}
\begin{eqnarray*}
0 \leq r \leq \gamma,\quad 0 \leq \psi_1 \leq 2\pi, \quad
-\frac{\pi}{2} \leq \psi_2\leq \frac{\pi}{2},\, \ldots,\,
-\frac{\pi}{2} \leq
\psi_{d} \leq \frac{\pi}{2}\\
\end{eqnarray*}

we obtain
\begin{gather}\label{pppp}
G_1(K, z)=\dfrac{1} {\sqrt{\cos \dfrac{K_1}{2}\ldots \cos
\dfrac{K_d}{2}}}\times\\ \times \int \limits_{0}^{\gamma} \int
\limits_{-{\pi}/{2}}^{{\pi}/{2}} \int
\limits_{-{\pi}/{2}}^{{\pi}/{2}} \ldots \int \limits_{0}^{2\pi}
\dfrac{r^{d-1}\omega(\psi)}{z- \cE_{\max}(K)+r^2}
\dfrac{dr\,d\psi_{d-1}\ldots
d\psi_1}{\sqrt{1-\cfrac{r^2\omega_1^2(\psi)}{2\cos\cfrac
{K_1}{2}}}\ldots \sqrt{1-\cfrac{r^2\omega_d^2(\psi)}{{2\cos\cfrac
{K_d}{2}}}}},\nonumber
\end{gather}
where $$\omega(\psi)=\omega(\psi_1,\ldots,\psi_{d-1}) =
\cos^{d-2}\psi_{d-1} \ldots \cos \psi_{2},$$
$$\omega_1(\psi)=\omega(\psi_1,\ldots,\psi_{d-1}) =\cos \psi_1
\cos \psi_2 \ldots \cos\psi_{d-1}$$ and
$$\omega_j(\psi)=\omega(\psi_1,\ldots,\psi_{d-1}) =\cos \psi_1
\cos \psi_2 \ldots \sin \psi_{d-j+1},\;j=2,\ldots,d.$$

Since
\begin{equation*}
\frac{1}{\sqrt{1-\cfrac{x^2}{2}}}=1+\frac{1}{4}x^2+
\frac{3}{32}x^4+\ldots,\quad\text{if}\hspace{2mm} |x|<\sqrt{2},
\end{equation*}
it is easy to see that \begin{equation}\label{1taqsim ildiz}
\frac{1}{\sqrt{1-\cfrac{r^2\omega_j^2(\psi)}{2\cos\frac
{K_j}{2}}}}=1+\frac{1}{4}\,\cfrac{\omega_j(\psi)}{\sqrt{\cos\frac
{K_j}{2}}}\,r^2 + \frac{3}{32}\,\cfrac{\omega_j^2(\psi)}{\cos\frac
{K_j}{2}}\,r^4+\ldots,\quad j=1,\ldots,d.\end{equation} Observe
that these series converges uniformly in any compact set ${\bf
K}\subset \Pi_0$ for any $r\in[0,\gamma].$ Therefore the series
obtained by term by term multiplication of these series
$$\frac{1}{\sqrt{1-\cfrac{r^2\omega_1^2(\psi)}{2\cos\frac
{K_1}{2}}}}\ldots\frac{1}{\sqrt{1-\cfrac{r^2\omega_d^2(\psi)}{2\cos\frac
{K_d}{2}}}}=1+\tilde{J}_2(\psi; K)r^2+ \tilde{J}_4(\psi;
K)r^4+\ldots $$ also converges uniformly in the set ${\bf K}\times
[0,\gamma].$

Set
\begin{equation}\label{wwww}
C_s(K)=\int \limits_{-{\pi}/{2}}^{{\pi}/{2}} \int
\limits_{-{\pi}/{2}}^{{\pi}/{2}} \ldots \int
\limits_{0}^{2\pi}\omega(\psi)\tilde{J}_s(\psi; K)d\psi_1 \ldots
d\psi_{d-1},
\end{equation}
$s=0,2,4,\ldots.$ Note that these functions are analytic functions
of $K$ on $\Pi_0$ since they depend on only
$$\left({\cos\frac{K_1}{2}}\right)^{-1/2},\ldots,
\left({\cos\frac{K_d}{2}}\right)^{-1/2}.$$ It is easy to see that
the series $$\sum\limits_{s=0}^{\infty} C_{2s}(K)r^{2s}$$ is
obtained by integrating uniformly convergent series and so
converges uniformly on ${\bf K}\times [0,\gamma].$

Therefore using the expansion of $\cos\frac{K}{2}$ for sufficiently
small neighborhood of $K=0$ one can find the following expansion:
\begin{equation} \label{C_s} \frac{C_s(K)}{\sqrt{\cos
\dfrac{K_1}{2}\ldots \cos
\dfrac{K_d}{2}}}=c_{s,0}+\sum\limits_{l=1}^d c_{s,1,l}K_{l}^2 +
\sum\limits_{l,m=1}^d c_{s,2,l,m}K_l^2K_m^2+\ldots,
\end{equation}
where  $c_{s,\ldots}$ are some real numbers. From this expansion we
establish that
\begin{equation*}
\Phi_0(K)=\frac{c}{\sqrt{\cos \dfrac{K_1}{2}\ldots \cos
\dfrac{K_d}{2}}}=c_{0,0}+O(|K|^2),\;\text{as}\; K\to 0,
\end{equation*}
where $c_{0,0}\ne 0.$

So the function $G_1(K, z)$ can be represented as
\begin{gather}\label{sss}
G_1(K,z)=\dfrac{1} {\sqrt{\cos \dfrac{K_1}{2}\ldots \cos
\dfrac{K_d}{2}}} \Bigg(C_0(K) \hspace{1mm}
\int\limits_{0}^{\gamma}\dfrac{r^{d-1}dr}{z- \cE_{\max}(K)+r^2}
\hspace{2mm}+ \\
C_2(K) \hspace{1mm} \int\limits_{0}^{\gamma}\dfrac{r^{d+1}dr}{z-
\cE_{\max}(K)+r^2} \hspace{2mm}+ C_4(K)\hspace{1mm}
\int\limits_{0}^{\gamma}\dfrac{r^{d+3}dr}{z- \cE_{\max}(K)+r^2}
\hspace{2mm}+ \hspace{2mm}\ldots\Bigg).\nonumber
\end{gather}
\vspace{3mm}

As above said, if $d\ge 3$ for any $K_0\in\Pi_0$ in the sufficiently
small compact neighborhood of $K=K_0$ and for all  $r$ small the
series $$\sum\limits_{s=0}^{\infty} C_{2s}(K)\,r^{2s+d-3}$$
converges uniformly. Therefore there exists
$\lim\limits_{z\to\cE_{\max}(K)}\nu(K,z)$ and
\begin{gather}\label{nyuk}
\nu(K)=\nu(K,\cE_{\max}(K))=\lim\limits_{z\to\cE_{\max}(K)}\nu(K,z)=\\
=\dfrac{1} {\sqrt{\cos \dfrac{K_1}{2}\ldots \cos \dfrac{K_d}{2}}}
\sum\limits_{s=0}^{\infty} C_s(K)\,\,\int\limits_0^{\gamma}
r^{2s+d-3} \,dr +
\int\limits_{\T^d\setminus{U_{\delta}{(\vec{\pi})}}}
\frac{dq}{\cE_{\max}(K)-\cE_K(q)}.\nonumber
\end{gather}
Since latter functions are analytic at any $K\in\Pi_0$ it follows
that $\nu(K)$ is analytic on $\Pi_0.$ $\blacktriangle$

Note that for $K\in U_{\delta}(0),$ using symmetricalness of
$\nu(\cdot)$ we obtain that
\begin{equation*}
\nu(K)=\nu(0)+a_2\sum\limits_{l=1}^d K_l^2+\sum\limits_{l,m=1}^d
a_{4,l,m} K_l^2K_m^2 +\ldots,
\end{equation*} where $a_{4,s,m}$ are real numbers,
and $a_2=\cfrac{\partial^2\nu(0)}{\partial K_1^2}.$

Part (i) of Lemma \ref{asimp} is proved. \vspace{4mm}

(ii) According to \eqref{pppp}, \eqref{1taqsim ildiz} and
\eqref{wwww}, the coefficient $C_0(K)$ does not depend on $K.$
Moreover if $d=1$ then $C_0(K)=2.$ If $d>1$ then using relation
\eqref{www} (see Appendix \ref{Appen b}) we get
$$C_0(K)=\int \limits_{-{\pi}/{2}}^{{\pi}/{2}} \int
\limits_{-{\pi}/{2}}^{{\pi}/{2}} \ldots \int
\limits_{0}^{2\pi}\cos^{d-2}\psi_{d-1} \ldots \cos \psi_{2}\,
d\psi_1 \ldots d\psi_{d-1}=$$ $$=2\pi\,\int
\limits_{-{\pi}/{2}}^{{\pi}/{2}} \cos^{d-2}\psi_{d-1}\,d\psi_{d-1}
\ldots \int \limits_{-{\pi}/{2}}^{{\pi}/{2}}\cos
\psi_{2}\,d\psi_{2}=2\pi \,a_{d-2}\ldots a_{1}.$$

Therefore it can be represented as
\begin{equation*}
C_0(K)=\begin{cases} 2,\quad &\mathrm{if}\quad  d=1\\
\dfrac{\left(2\pi\right)^{m+1}}{(2m)!!},\hspace{2mm}&
\mathrm{if}\hspace{2mm}  d=2m+2\\ \dfrac{
\left(2\pi\right)^{m+1}}{(2m-1)!!},\hspace{2mm} &
\mathrm{if}\hspace{2mm} d=2m+1\\
\end{cases}
\end{equation*}
\vspace{3mm}

Using Lemma \ref{1 taqsim r^2-z} (Appendix \ref{Appen b}) and
relations \eqref{sss} and \eqref{kkk}, we get the expansion for the
function $\nu(K,\cdot)$ at the point $z=\cE_{\max}(K)$
\begin{eqnarray*}
\nu(K, z)&= &-\dfrac{\Phi_0(K)}{2} \left(\cE_{\max}(K)-z\right)^m
\ln(z-\cE_{\max}(K))+\\&+ &\left(\cE_{\max}(K)-z\right)^{m+1}
\ln(z-\cE_{\max}(K))\Phi_{11}(K,z)+\widetilde{\Phi}_2(K,z),\end{eqnarray*}
if $d=2m+2$ and
\begin{eqnarray*}
\nu(K, z)= \dfrac{\pi \Phi_0(K)}{2}\hspace{1mm}
\dfrac{\left(\cE_{\max}(K)-z\right)^{m}}{\sqrt{z-\cE_{\max}(K)}}
+\left(\cE_{\max}(K)-z\right)^{m+1/2} \Phi_{12}(K,z)+
\widetilde{\Phi}_2(K,z),
\end{eqnarray*} if $d=2m+1,$
where $z>\cE_{\max}(K)$ and $$\Phi_0(K)=\dfrac{C_0(K)}{\sqrt{\cos
\dfrac{K_1}{2}\ldots\cos \dfrac{K_d}{2}}}$$ and
$\Phi_{11}(K,\cdot)$, $\widetilde{\Phi}_{2}(K,\cdot)$ are regular
functions in some neighborhood $(\cE_{\max}(K),\cE_{\max}(K)+\xi)$
of $z=\cE_{\max}(K).$ Since
$\nu(K)=\widetilde{\Phi}_2(K,\cE_{\max}(K)),$ by the regularity of
$\widetilde{\Phi}_2(K,z)$ we can rewrite $\widetilde{\Phi}_2(K,z)$
as following: $$\widetilde{\Phi}_2(K,z)=\nu(K)+\Phi_2(K,z),$$
where $\Phi_2(K,\cdot)$ is also regular function in
$(\cE_{\max}(K),\infty)$ and $\Phi_2(K,\cE_{\max}(K))=0.$
\end{proof}
\vspace{3mm}

\begin{corollary}
The function ${\nu}(K, z)$ can be analytically continued to some
punctured  $\delta_1$ - neighborhood $\dot
V_{\delta_1}(\cE_{\max}(K))$ of $\cE_{\max}(K)$ as
\begin{eqnarray*}
\nu^*(K, z)& = & \widetilde \Phi_{11}(K,z)\,\mathrm{
Ln}(z-\cE_{\max}(K))+ \Phi_2(K,z),\quad\text{if\,
$d=2m+2$}\end{eqnarray*}  and
\begin{eqnarray*}
\nu^*(K, z)& = & \sum\limits_{l=0}^{\infty} b_l(K)\,
\left(z-\cE_{\max}(K)\right)^{(m+l)/2} +\Phi_2(K,z), \quad\text{if
$d=2m+1$}.
\end{eqnarray*}
\end{corollary}

\begin{proof}[\bf Proof]
Note that the proposition of Lemma \ref{asimp} is still hold if we
change a real half neighborhood $(\cE_{\max}(K),\cE_{\max}(K)+\xi)$
to a punctured complex neighborhood $\dot
V_{\delta_1}(\cE_{\max}(K))$ of $\cE_{\max}(K).$ Therefore the
function ${\nu}(K, \cdot)$ may be analytically continued to $\dot
V_{\delta_1}(\cE_{\max}(K)).$
\end{proof}

\begin{remark}
The proof of the Lemma yields that in case $d=1$ the function
$\nu(K,z)$ has precise form:
$$\nu(K, z)=\dfrac{\pi}{\sqrt{\cos \dfrac{K}{2}}}
\dfrac{1}{\sqrt{z-\cE_{\max}(K)}}-\dfrac{\pi}{8\sqrt{\cos^3
\dfrac{K}{2}}} \sqrt{z-\cE_{\max}(K)} + $$
$$+\dfrac{3\pi}{128\sqrt{\cos^5 \dfrac{K}{2}}}
\sqrt{\left(z-\cE_{\max}(K)\right)^3}-\dfrac{5\pi}{1024\sqrt{\cos^7
\dfrac{K}{2}}} \sqrt{\left(z-\cE_{\max}(K)\right)^5}+...$$
\end{remark}

\begin{proof}[\bf Proof of Theorem \ref{assymptotica}] The existence of
unique eigenvalue $z(\mu,K)$ is a simple consequence of the
Birman-Schwinger principle and the Fredholm theorem.

Indeed, for any fixed $K\in \Pi_0$  the function
${\Delta}_\mu(K,\cdot)$ is analytic and monotone increasing  in
$(\cE_{\max}(K),\infty).$ Moreover for $\mu>\mu_0(K)$
\begin{equation}\label{properties of det}
\lim_{z\to {+\infty}}\Delta_\mu(K,z)=1\,\, \mbox{and}
 \,\,\lim_{z\to {\tiny
 \cE_{\max}(K)}}\Delta_\mu(K,z)=1-\frac{\mu}{\mu_0(K)}<0.
\end{equation}
Hence there is a unique number $z(\mu,K)\in (\cE_{\max}(K),\infty)$
such that $${\Delta}_\mu(K,z(\mu,K))=0.$$ According to Lemma
\ref{eigen},  $z(\mu,K)$ is an eigenvalue of the operator
$H_{\mu}(K).${\vspace{3mm}}

Let $\mu(K,z)=[\nu(K,z)]^{-1}.$ The function $\mu(K,\cdot)$ is
monotone and regular on $(\cE_{\max}(K)),+\infty).$ Since $d\ge 3,$
there exists finite limit
$$\lim\limits_{z\to\cE_{\max}(K)+0}\mu(K,z)=\mu(K,\cE_{\max}(K)),$$
hence, we can redefine this function at the point
$\cE_{\max}(K)\in\R$ as $\mu(K,\cE_{\max}(K))=\mu_0(K).$ By the
theorem on the existence of inverse monotone function, there exists
such inverse function $z(\cdot,K):[\mu_0(K),+\infty)\to
[\cE_{\max}(K)),+\infty)$ that the relation
$$\Delta_{\mu}(K,z(\mu,K))=0$$ holds for all
$\mu\in[\mu_0(K),+\infty).$ By the regularity of the function
$\mu(K,\cdot)$ in $(\cE_{\max}(K)),+\infty)$ it follows that
$z({\cdot,K})$ is also regular in $(\mu_0(K),+\infty).$ By the
continuity and monotonicity of the functions $\mu(K,\cdot)$ and
$z({\cdot,K})$ at $z=\cE_{\max}(K)$ and $\mu=\mu_0(K),$
respectively, we get $\mu\to\mu_0(K)$ if and only if
$z\to\cE_{\max}(K).$\vspace{5mm}

The proves of parts (i)-(iv) of Theorem \ref{assymptotica} are based
on Lemma \ref{asimp}.

Taking into account
$$\frac{1}{\mu(K,z)}-\frac{1}{\mu_0(K)} = \nu(K,z)-\nu(K)$$ we get
\begin{gather}\label{mu0}
-\frac{\mu(K,z)-\mu_0(K)}{\mu_0(K)(\mu(K,z)-\mu_0(K))+(\mu_0(K))^2}
= {\nu(K,z)-\nu(K)}. \end{gather}\vspace{5mm}

(i). Let $d=3.$ By  Lemma \ref{asimp} we can rewrite \eqref{mu0} as
following:
\begin{eqnarray}\label{asss} -\frac{\lambda}{\mu_0\lambda + \mu_0^2}
=f(\alpha, K),
\end{eqnarray}
where $\lambda =\mu(K,z)-\mu_0(K),$ $\alpha=
(z-\cE_{\max}(K))^{1/2},$ $\mu_0=\mu_0(K)$ and $f(\alpha, K)$ is
regular function in some neighborhood of $\alpha=0$ with
$$f(0, K)=0,\quad\dfrac{\partial f}{\partial \alpha} (0,K) = -\dfrac{\pi
\Phi_0(K)}{2}<0.$$

The proof of  part ${\rm (i)}$ of Theorem \ref{assymptotica} follows
from Theorem \ref{f^-1} (see Appendix \ref{Appen a}). \vspace{5mm}

(ii). Let $d=4.$ By Lemma \ref{asimp}, we can rewrite \eqref{mu0} as
following: $$\lambda = -(\mu_0 \lambda +\mu_0^2) (f(\alpha,
K)+g(\alpha, K)\ln \alpha)$$ where $\lambda=\mu(K,z)-\mu_0(K),$
$\alpha= z-\cE_{\max}(K),$ $\mu_0=\mu_0(K)$ and $f(\cdot, K),$
$g(\cdot, K)$ are regular functions in some neighborhood of
$\alpha=0$ with $$f(0, K)=g(0, K)=0, \quad g'(0,
K)=\frac{\Phi_0(K)}{2}>0.$$

The proof of this part of the theorem follows from Theorem \ref{d=4}
(see Appendix \ref{Appen a}). \vspace{5mm}

(iii). Let $d\ge 5$ and odd. Since $K\in \Pi_0$ then $\nu(K,\cdot)$
is differentiable in $[\cE_{\max}(K),+\infty)$ and
$$\frac{\partial \nu}{\partial z}(K,\cE_{\max}(K))=-\int\limits_{\T^d}
\frac{dq}{(\cE_{\max}(K)-\cE_{K}(q))^2}< 0$$ for all $K\in\Pi_0.$
Therefore, by Lemma \ref{asimp} we get
$$\nu(K,z)=\nu(K)+\frac{\partial \nu}{\partial z}(K,
\cE_{\max}(K))(z-\cE_{\max}(K))+ \sum\limits_{n\ge 4}
(z-\cE_{\max}(K))^{n/2}.$$  We can rewrite \eqref{mu0} as following:
$$-\frac{\lambda}{\mu_0 \lambda+\mu_0^2} = f(\alpha, K),$$ where
\be\label{hoji}\lambda=\mu(K,z) - \mu_0(K),\quad \alpha =
(z-\cE_{\max}(K))^{1/2},\quad \mu_0=\mu_0(K)\ee and $f(\cdot,K)$ is
regular function in some neighborhood of $\alpha=0$ with
$$f(0, K)=0,\quad\frac{\partial f}{\partial \alpha}(0, K)=0,\quad
\frac{\partial^2 f}{\partial \alpha^2}(0, K)=-2\int\limits_{\T^d}
\frac{dq}{(\cE_{\max}(K)-\cE_{K}(q))^2}< 0.$$

The proof of part ${\rm (iii)}$ of the theorem follows from Theorem
\ref{d katta 5} (see Appendix \ref{Appen a}). \vspace{5mm}

(iv). Let $d\ge 6$ and even (if $d=2m+2$ then $m\ge 2$). Note that
in this case $\nu(K,z)$ is represented as
$$\nu(K,z)=\nu(K)+\sum\limits_{n\ge 1} A_n(K) (z-\cE_{\max}(K))^n -
$$$$-\dfrac{\Phi_0(K)}{2} \left(\cE_{\max}(K)-z\right)^m
\ln(z-\cE_{\max}(K))\sum\limits_{n\ge 1} B_n(K)
(z-\cE_{\max}(K))^n.$$ Then equation \eqref{mu0} takes the form
$$-\dfrac{\lambda}{\mu_0\lambda + \mu_0^2} = f(\alpha, K)+
g(\alpha,K)\, \ln \alpha ,$$ where \be\label{toji}\lambda=\mu(K,z) -
\mu_0(K),\quad \alpha = (z-\cE_{\max}(K))^{1/2},\quad
\mu_0=\mu_0(K)\ee and functions $f(\cdot,K),$ $g(\cdot,K)$ are
regular functions in the some neighborhood of $\alpha=0$ with
$$f(0,K)=\dfrac{\partial f}{\partial \alpha}(0,K)=0,\quad A_1(K)=
\frac{\partial^2 f}{\partial \alpha^2}(0, K)=-2\int\limits_{\T^d}
\frac{dq}{(\cE_{\max}(K)-\cE_{K}(q))^2}< 0,$$
$$g(0,K)=\dfrac{\partial g}{\partial \alpha}(0,K)=
\frac{\partial^2 g}{\partial \alpha^2}(0, K)=0.$$

The proof of part ${\rm (iv)}$  of the theorem follows from Theorem
\ref{d kat 6} (see Appendix \ref{Appen a}).
\end{proof}

\vspace{5mm}

\begin{proof}[\bf Proof of Theorem \ref{quasi}.] Note that for any $K\in
\Pi_0\setminus\{0\}$ the inequality $\mu_0(K)<\mu_0(0)$ holds.
Therefore the first part of Theorem \ref{assymptotica} implies that
$H_{\mu_0(0)}(K)$ has a unique eigenvalue $z(\mu_0(0),K)$ above the
essential spectrum $\sigma_{ess}(H_{\mu_0(0)}(K)).$

Since the function $\nu(K,z)$ is real-analytic in $\T^d\times
\R\setminus [\cE_{\min}(K),\cE_{\max}(K)]$ and the function
$z(\mu_0(K),K)$ is solution of the equation
$\Delta_{\mu_0(K)}(K,z)=0,$ the implicit function theorem implies
that $z(\mu_0(0),K)$ is real-analytic function in
$\Pi_0\subset\T^d.$ Observe that for the eigenvalue of
$H_{\mu_0(0)}(K)$ the relation $z(\mu_0(0),K)\to \cE_{\max}(0)$
holds as $K\to 0.$

According to the proof of Lemma \ref{asimp}, the functions
$\nu(K)$ and $\Phi_0(K)$ can be written as
\begin{equation}\label{nu_k}
\nu(K)=\nu(0)+a_2\sum\limits_{l=1}^d K_l^2+\sum\limits_{l,m=1}^d
a_{4,l,m} K_l^2K_m^2 +O(|K|^4),
\end{equation} with $a_2=\cfrac{\partial^2\nu(0)}{\partial K_1^2}$
and
\begin{equation}\label{fi0}
\Phi_0(K)=\frac{c}{\sqrt{\cos \dfrac{K_1}{2}\ldots \cos
\dfrac{K_d}{2}}}=c_{0,0}+O(|K|^2),\,K\to 0,\quad c_{0,0}\ne
0.\end{equation}

${(\rm i)}$ Let $z_0(K)=z({\mu_0(0)},K)$ is eigenvalue of
$H_{\mu_0(0)}(K).$ Then it is clear that
\begin{equation}\label{deter}\Delta_{\mu_0(0)}(K,z_0(K))
=1-\mu_0(0)\int\limits_{\T^d}
\dfrac{dq}{z_0(K)-\cE_K(q)}=0.\end{equation} By Lemma \ref{asimp}
the following relation $$0=1-\mu_0(0)\Big(\dfrac{\pi \Phi_0(K)}{2}
(z_0(K)-\cE_{\max}(K))^{1/2} +$$ $$
\nu(K)+(z_0(K)-\cE_{\max}(K))O(1)\Big),\, K\to 0$$ holds. Since
$\nu(0)=(\mu_0(0))^{-1}$ we get
$$(z_0(K)-\cE_{\max}(K))^{1/2}\left[\dfrac{\pi \Phi_0(K)}{2}+(z_0(K)-
\cE_{\max}(K))^{1/2}O(1)\right]=-(\nu(K)-\nu(0))$$ and hence
\eqref{fi0} and \eqref{nu_k} imply that
$$(z_0(K)-\cE_{\max}(K))\left[c_{0,0}+O(|K|^2)+(z_0(K)-
\cE_{\max}(K))^{1/2}O(1)\right]^2=$$$$=|K|^4(a_2+a_4|K|^2+o(|K|^2))^2,\;
K\to 0.$$  According to the regularity of $z_0(K)$ on $\Pi_0$ and
$c_{0,0}\ne 0$ we get that
$$z_0(K)-\cE_{\max}(K)=c_{0,0}^{-2} a_2^2|K|^4+o(|K|^4), \;\text{as}\; K\to 0.$$
Using $\cE_{\max}(K)= \cE_{\max}(0)-\frac14|K|^2+O(|K|^4),\; K\to
0,$ we establish
$$z({\mu_0(0)},K)=\cE_{\max}(0)-\frac14|K|^2+ O(|K|^4),\; K\to 0.$$

\vspace{5mm}

${\rm (ii)}$ Taking into account \eqref{nu_k}, using  Lemma
\ref{asimp} and the fact that $z_0(K)=z(\mu_0(0),K)$ is a solution
of $\Delta_{\mu_0(0)}(K,z)=0,$ we rewrite \eqref{deter} as
following:
$$(z_0(K)-\cE_{\max}(K))\ln[z_0(K)-\cE_{\max}(K)]+(z_0(K)-\cE_{\max}(K))\,
O(1) =$$ $$ |K|^2(c+O(|K|^2)),\,K\to 0$$ where $c\ne 0.$ Since
$$0=\lim\limits_{K\to 0}\frac{z_0(K)-\cE_{\max}(K)}{(z_0(K) -\cE_{\max}(K))
\ln(z_0(K)-\cE_{\max}(K))+(z_0(K)-\cE_{\max}(K))\, O(1)} =$$$$=
\lim\limits_{K\to 0}\frac{z_0(K)
-\cE_{\max}(K)}{|K|^2(c+O(|K|^2))}$$ we obtain
$z_0(K)-\cE_{\max}(K)=o(|K|^2),$ as $K\to 0.$ Consequently
$$z(\mu_0(0),K)-\cE_{\max}(0)=-\frac14|K|^2+o(|K|^{2}),\;K\to0.$$
\vspace{5mm}

${\rm (iii)}$ It is easy to get that
\begin{equation}\label{k to 0}
\frac{\partial \nu (K,\cE_{\max}(K))}{\partial z}= \frac{\partial
\nu (0,\cE_{\max}(0))}{\partial z}+O(|K|),\quad K\to 0.
\end{equation}

Using Taylor expansion of $\nu(K,z),$ we obtain
$$\nu(K,z)=\nu(K)+\frac{\partial \nu (K,\cE_{\max}(K))}{\partial
z}(z-\cE_{\max}(K))+ (z-\cE_{\max}(K))o(1),$$ as $z\to\cE_{\max}(K)$
and taking into account that $z_0(K)=z(\mu_0(0),K)$ is solution of
the equation $\Delta_{\mu_0(0)}(K,z)=0,$ we get
$$-(\nu(K)-\nu(0))=\dfrac{\partial\nu (K, \cE_{\max}(K))}{\partial
z}(z_0(K)-\cE_{\max}(K))+ (z_0(K)-\cE_{\max}(K))o(1),$$ as $K\to 0,$
which means that
$$z_0(K)-\cE_{\max}(K)=\frac{a_2|K|^2 +
o(|K|^2)}{-\dfrac{\partial\nu}{\partial z}(0, \cE_{\max}(0))+
O(|K|)}=\frac{a_2}{-\dfrac{\partial \nu}{\partial
z}(0,\cE_{\max}(0))}|K|^2+o(|K|^2).$$ Consequently
$$z(\mu_0(0),K)-\cE_{\max}(0) =
\left(\frac{a_2}{-\dfrac{\partial\nu}{\partial z}(0,\cE_{\max}(0))}
- \frac14\right)|K|^2+o(|K|^2),\;K\to 0,$$ where
$a_2=\cfrac{\partial^2\nu(0)}{\partial K_1^2}.$
\end{proof}

\appendix
\section{Consequences of the Implicit Function Theorem}\label{Appen a}

In this appendix we want to prove results which are used in Theorem
\ref{assymptotica}. Such kinds of lemmas are consi\-dered in
\cite{MK.BS}. By making suitable substitutions, we will reduce the
proof to the implicit function theorem in several variables.

\vspace{3mm}

\begin{theorem}\label{f^-1}
Let $f(\alpha)$ be analytic function in a neighborhood of
$\alpha=0.$ Suppose that $$f(0)=0, \quad f'(0) < 0.$$ Then for
$\lambda$ sufficiently small there is a unique positive
$\alpha(\lambda)$ satisfying \begin{equation}\label{eshmat}
\lambda=-(\mu_0\lambda+\mu_0^2) f(\alpha).
\end{equation} Moreover, for $\lambda$ sufficiently small, $\alpha$
has a convergent expansion \begin{equation} \alpha =
\sum\limits_{n\ge 1} c_n \lambda^n,
\end{equation} with $c_1=-[\mu_0^2f'(0)]^{-1}.$
\end{theorem}

\begin{proof}[\bf Proof]
We can write $$f(\alpha)=\sum\limits_{n\ge 1} a_n \alpha^n,\quad
a_1= f'(0)<0.$$ In \eqref{eshmat} try substitution
$$\alpha=\lambda(c+\xi),\quad c=-(a_1\,\mu_0^2)^{-1}.$$ Then in the
region where $|\alpha|$ sufficiently small this is equivalent to
\begin{equation}\label{abdulla}
a_1\mu_0^2 \xi=-(c+\xi)\left\{a_1\mu_0\lambda  + (\mu_0\lambda +
\mu_0^2) \sum\limits_{n\ge 2} a_n \lambda^{n-1}
(c+\xi)^{n-1}\right\}
\end{equation}

Equation \eqref{abdulla} can be written in the form $$F(\xi,
\lambda) = 0 $$ where (i) $\xi=0,$ $\lambda=0$  is a solution; (ii)
$F$ is analytic for $|\lambda|,$ $|\xi|$ small; (iii) $\partial F/
\partial \xi (0,0) = -a_1\mu_0^2\ne 0.$ Thus by the implicit
function theorem, \eqref{abdulla} has a unique solution for $\xi,$
$\lambda$ small given by a convergent expansion $$\xi =
\sum\limits_{n\ge0} b_n \lambda^n.$$ Consequently, $$\alpha =
\lambda(c+\xi) = \sum\limits_{n\ge 1} c_n \lambda^n,$$ where $c_1=c=
-[\mu_0^2\,f'(0)]^{-1}.$
\end{proof}

\begin{theorem}\label{d=4}
Let $f(\alpha),$ $g(\alpha)$ be analytic functions in a neighborhood
of $\alpha=0.$ Suppose that $$f(0)=g(0)=0,\quad g'(0)> 0.$$ Then for
$\lambda$ sufficiently small and positive there is a unique positive
$\alpha(\lambda)$ satisfying
\begin{equation}\label{l=f+g} \lambda
= - (\mu_0\lambda + \mu_0^2) \left(f(\alpha) + g(\alpha) \ln
\alpha\right).\end{equation} Moreover, for $\lambda$ sufficiently
small, $\alpha$ has a convergent expansion
\begin{equation}\label{a=f+g} \alpha(\lambda) = \sum\limits_{n\ge1,\,m,k,l\ge
0} c(n,m,k,l)\; \sigma^n \tau^m \omega^k \lambda^l \end{equation}
with \be\label{ttau}\sigma=\frac{\lambda}{-\ln \lambda},\quad
\tau=\frac{1}{-\ln \lambda}, \quad \omega=\frac{\ln \ln
\lambda^{-1}}{-\ln \lambda}\ee and $c(1,0,0,0)=(\mu_0^2\,\,\,
g'(0))^{-1}.$
\end{theorem}

\begin{proof}[\bf Proof]
We write $$f(\alpha)=\sum\limits_{n\ge 1} a_n \alpha^n, $$
$$g(\alpha)=\sum\limits_{n\ge 1} b_n \alpha^n,\quad b_1=g'(0)>0.$$ Try the
substitution \be\label{s(1+z)}\alpha=\sigma\,(c+\xi),\quad
c=(b_1\,\mu_0^2)^{-1}.\ee Then, it is easy to find that in the
region where $|\alpha|$ is small, \eqref{l=f+g} is equivalent to
\begin{gather}\label{l=tau}
b_1\,\mu_0^2\,\xi=(c+\xi)\Bigg\{-\mu_0\lambda b_1+
(\mu_0\lambda+\mu_0^2)\Big[\tau \sum\limits_{n\ge0} a_{n+1} \sigma^n
(c+\xi)^n- \\- \sum\limits_{n\ge1} b_{n+1} \sigma^n (c+\xi)^n  -
\omega \sum\limits_{n\ge0} b_{n+1} \sigma^n(c+\xi)^n - \tau\ln
(c+\xi) \sum\limits_{n\ge 0} b_{n+1} \sigma^n (c+\xi)^n
\Big]\Bigg\}\nonumber.
\end{gather}

Equation \eqref{l=tau} can be written in the form
$$F(\xi,\sigma,\tau,\omega,\lambda)=0,$$ where $\rm(i)$ $\xi=0,$ $\sigma=0,$
$\tau=0,$ $\omega=0,$ $\lambda=0$ is a solution; $\rm (ii)$ $F$ is
analytic for $|\xi|,$ $|\sigma|,$ $|\tau|,$ $|\omega|,$ $|\lambda|$
small since $\ln (c+\xi)$ is analytic at $\xi=0;$ $(\rm iii)$
$\partial F/\partial \xi (0,0,0,0,0)=b_1\mu_0^2\ne 0.$ Thus by the
implicit function theorem, \eqref{l=tau} has a unique solution for
$\xi,$ $\sigma,$ $\tau,$ $\omega,$ $\lambda$ small given by a
convergent expansion $$\xi=\sum\limits_{n,m,k,l\ge 0} d(n,m,k,l)\,
\sigma^n\,\tau^m\,\omega^k\,\lambda^l.$$ Given \eqref{s(1+z)} and
\eqref{ttau}, this yields \eqref{a=f+g}.
\end{proof}

\begin{theorem}\label{d katta 5}
Let $f(\alpha)$ be analytic function in a neighborhood of
$\alpha=0.$ Suppose that $$f(0)=f'(0)=0, \quad f''(0) < 0.$$ Then
for $\lambda$ sufficiently small there is a unique positive
$\alpha(\lambda)$ satisfying \begin{equation}\label{eshmat}
\lambda=-(\mu_0\lambda+\mu_0^2) f(\alpha).
\end{equation} Moreover, for $\lambda$ sufficiently small, $\alpha$
has a convergent expansion \begin{equation} \alpha = \left(
\sum\limits_{n\ge 1} c_n \lambda^{n/2}\right)^2,
\end{equation} with $c_1=[-1/2\mu_0^2\, f''(0)]^{-1/2}.$
\end{theorem}

\begin{proof}[\bf Proof]
We can write $$f(\alpha)=\sum\limits_{n\ge 2} a_n \alpha^n,\quad
a_2=1/2f''(0)<0.$$ In \eqref{eshmat} try substitution
$$\alpha=\sigma(c+\xi),\quad c=(-a_2\,\mu_0^2)^{-1/2},$$ where
$\sigma=\lambda^{1/2}.$ Then in the region where $|\alpha|$
sufficiently small this is equivalent to
\begin{equation}\label{aabdulla}
2\mu_0\sqrt{-a_2}\, \xi - a_2 \mu_0^2\, \xi^2 =(c+\xi)^2
\left\{\mu_0a_2 \sigma^2 - (\mu_0\sigma^2 + \mu_0^2)
\sum\limits_{n\ge 3} a_n \sigma^{n-2} (c+\xi)^{n-2}\right\}
\end{equation}

Equation \eqref{aabdulla} can be written in the form $$F(\xi,
\sigma) = 0$$ where (i) $\xi=0,$ $\sigma=0$  is a solution; (ii) $F$
is analytic for $|\sigma|,$ $|\xi|$ small; (iii) $\partial F/
\partial \xi (0,0) = 2\mu_0\sqrt{-a_2}\ne 0.$ Thus by the implicit
function theorem, \eqref{aabdulla} has a unique solution for $\xi,$
$\sigma$ small given by a convergent expansion $$\xi =
\sum\limits_{n\ge0} b_n \sigma^n.$$ Consequently, $$\alpha =
\sigma(c+\xi) = \sum\limits_{n\ge 1} c_n \sigma^n,$$ where $c_1=c=
[-1/2\mu_0^2\,f''(0)]^{-1/2}.$
\end{proof}

\begin{theorem}\label{d kat 6}
Let $f(\alpha),$ $g(\alpha)$ be analytic functions in a neighborhood
of $\alpha=0.$ Suppose that $$f(0)=f'(0)=g(0)=g'(0)=g''(0)=0,\quad
f''(0)<0,\quad .$$ Then for $\lambda$ sufficiently small and
positive there is a unique positive $\alpha(\lambda)$ satisfying
\begin{equation}\label{l=f+g} \lambda
= - (\mu_0\lambda + \mu_0^2) \left(f(\alpha) + g(\alpha) \ln
\alpha\right).\end{equation} Moreover, for $\lambda$ sufficiently
small, $\alpha$ has a convergent expansion
\begin{equation}\label{a=f+g} \alpha(\lambda) = \sum\limits_{n\ge1,\,k\ge
0} c(n,k)\; \sigma^n \tau^k  \end{equation} with
\be\label{ttau}\sigma={\lambda}^{1/2},\quad \tau=\sigma \ln
\sigma\ee and $c(1,0)=(-1/2 \mu_0^2\,\,\, f''(0))^{-1/2}.$
\end{theorem}

\begin{proof}[\bf Proof]
We write $$f(\alpha)=\sum\limits_{n\ge 2} a_n \alpha^n,\quad a_2=1/2
f''(0)<0,\qquad g(\alpha)=\sum\limits_{n\ge 3} b_n \alpha^n .$$ Try
the substitution \be\label{s(1+z)} \lambda=\sigma^2,\quad
\alpha=\sigma\,(c+\xi),\quad c=(-a_2\,\mu_0^2)^{-1/2}.\ee Then
\begin{gather}\label{l=teau}
2\mu_0\sqrt{-a_2}\,\xi - \mu_0^2 a_2 \,
\xi^2=(c+\xi)^2\Bigg\{-\mu_0\sigma^2 a_2+
(\mu_0\sigma^2+\mu_0^2)\Big[\sum\limits_{n\ge3} a_{n} \sigma^{n-2}
(c+\xi)^{n-2}+ \\ + \tau\sum\limits_{n\ge3} b_{n} \sigma^{n-3}
(c+\xi)^{n-2} + \ln(c+\xi) \sum\limits_{n\ge3} b_{n}
\sigma^{n-2}(c+\xi)^{n-2} \Big]\Bigg\}\nonumber,
\end{gather} where $\tau=\sigma\ln \sigma.$

Equation \eqref{l=teau} can be written in the form
$$F(\xi,\sigma,\tau)=0,$$ where $\rm(i)$ $\xi=0,$ $\sigma=0,$
$\tau=0$ is a solution; $\rm (ii)$ $F$ is analytic for $|\xi|,$
$|\sigma|,$ $|\tau|$ small since $\ln (c+\xi)$ is analytic at
$\xi=0;$ $(\rm iii)$ $\partial F/\partial \xi (0,0,0)=2\mu_0
\sqrt{-a_2}\ne 0.$ Thus by the implicit function theorem,
\eqref{l=tau} has a unique solution for $\sigma,$ $\tau$ small given
by a convergent expansion $$\xi=\sum\limits_{n,k\ge 0} d(n,k)\,
\sigma^n\,\tau^k.$$ Given \eqref{s(1+z)} and \eqref{ttau}, this
yields \eqref{a=f+g}.
\end{proof}

\section{}\label{Appen b}

In this section we give two lemmas which are used in the proof of
Lemma \ref{asimp}.

\begin{lemma}\label{1 taqsim r^2-z}
The function $$I_s(\theta)=\int\limits_{0}^{\gamma}
\frac{r^s}{r^2-\theta}, \hspace{3mm} \theta<0,\,s=0,1,2,\ldots$$ is
represented as
\begin{equation}
I_s(\theta)=\begin{cases}
-1/2\,\theta^m\, \ln(-\theta)+\hat{I}_s(\theta), & \text{s=2m+1}\\
\dfrac{\pi}{2\sqrt{-\theta}}\,\theta^m + \tilde{I}_s(\theta), &
\text{s=2m}
\end{cases}
\end{equation}
where $\hat{I}_s(\theta)$ and $\tilde{I}_s(\theta)$ are regular
functions in the some neighborhood of zero.
\end{lemma}
\vspace{3mm}

\begin{proof}[\bf Proof]
We will prove for $s=2m+1.$ For the case $s=2m$ it can be proved
analogously. Integrating the identity
\begin{equation*}
\frac{r^{2m+1}}{r^2-\theta}=r\left(r^{2(m-1)}+\theta
r^{2(m-2)}+\ldots+\theta^{m-1}\right) +\frac{r\theta^m}{r^2-\theta}
\end{equation*}
from $0$ to $\gamma,$ we get
\begin{equation*}
\int\limits_{0}^{\gamma}\frac{r^{2m+1}}{r^2-\theta}=
\frac{\gamma^{2m}}{2m}+\theta\frac{\gamma^{2(m-1)}}{2(m-1)}+\ldots+
\theta^{m-1}\frac{\gamma^2}{2}+ \frac{1}{2}\, \theta^m\, \ln
\frac{\gamma^2-\theta}{-\theta}
\end{equation*}

Therefore
\begin{equation*}
\int\limits_{0}^{\gamma}\frac{r^{2m+1}}{r^2-\theta}=
-\frac{1}{2}\,\theta^m
\,\ln(-\theta)+\hat{p}_{m-1}(\theta)+f_1(\theta),
\end{equation*}

Denote by
\begin{align*}
\hat{I}_m(\theta)=\hat{p}_{m-1}(\theta)+f_1(\theta),
\end{align*}
where $\hat{p}_{m-1}(\theta)$ is a polynomial with degree $m-1$,
$f_1(\theta)$ is a regular function in some neighborhood of zero.
\end{proof}

\begin{lemma}
\begin{equation}\label{www}
a_n=\int\limits_{-{\pi}/{2}}^{{\pi}/{2}}\cos^{n}{t}dt=
\begin{cases}
\pi\dfrac{(2m-1)!!}{(2m)!!}, & n=2m\\
2\dfrac{(2m)!!}{(2m+1)!!}, & n=2m+1\\
\end{cases}
\end{equation}
\end{lemma}
\vspace{3mm}

\begin{proof}[\bf Proof]
Integrating by parts we easily receive the relation
$a_n=\dfrac{n-1}{n}a_{n-2},$ where $n=2,3,4,\ldots.$ Using
$a_0=\pi,$ $a_1=2$ we get \eqref{www}.
\end{proof}
\vspace{3mm}

This work was supported by DFG projects and the Fundamental Science
Foundation of Uzbekistan. The second named author gratefully
acknowledge to the hospitality of the University of Mainz.

\end{document}